\documentclass[11pt,reqno,a4paper]{amsart}
\usepackage{graphicx,epsfig}
\usepackage{mathrsfs}
\usepackage{amsaddr}
\usepackage{amsthm}
\usepackage{enumitem}
\usepackage{caption}
\usepackage[numbers,sort&compress]{natbib}
\usepackage{soul}
\usepackage{todonotes}
\usepackage{hyperref}
\usepackage[final,color,notref,notcite]{showkeys}
\usepackage[charter]{mathdesign}
\usepackage[scr=boondoxo, scrscaled=.98]{mathalfa}

\DeclareSymbolFont{Eulerscripteusm10}{U}{eus}{m}{n}
\SetSymbolFont{Eulerscripteusm10}{bold}{U}{eus}{b}{n}
\DeclareMathSymbol{\euM}{\mathord}{Eulerscripteusm10}{"4D}
\DeclareMathSymbol{\euD}{\mathord}{Eulerscripteusm10}{"44}

\DeclareMathAlphabet{\mathbxit}{\encodingdefault}{\rmdefault}{bx}{it}   % bold and italic
\newcommand{\ubi}{\mathbxit{u}}
\newcommand{\dbi}{\mathbxit{d}}
\newcommand{\Ubi}{\mathbxit{U}}
\newcommand{\Dbi}{\mathbxit{D}}

\DeclareMathAlphabet{\pazocal}{OMS}{zplm}{m}{n}   %  mathcal

\newcommand{\Ncal}{\pazocal{N}}

\newcommand{\Ocal}{\pazocal{O}}

\newtheorem*{prop}{Proposition}
\newtheorem*{lem}{Lemma}

\theoremstyle{remark}
\newtheorem*{rem}{\bfseries\textup{Remark}}
\newtheorem*{defi}{\bfseries\textup{Definition}}
\newtheorem{exa}{\bfseries\textup{Example}}

\definecolor{labelkey}{rgb}{0,.56,.7}

\newcommand*{\at}{@}
\newcommand{\papa}[2]{\frac{\partial#1}{\partial#2}}
\newcommand{\nn}{\nonumber}

\def\dg{\dagger}
\def\df{\overset{\mathrm{df}}{=}}

\newcommand{\id}{\mathop{{\mathrm{id}}}\nolimits}

\newcommand{\aadg}[1]{a_{#1}^{\dg}}
\newcommand{\adg}{{a^\dg}}         % a^dagger
\newcommand{\Adg}{{A^\dg}}         % A^dagger

\def\loD{\underline{\smash\euD}}
\def\hiD{\overline{\euD}}
\newcommand{\ket}[1]{\mathop{|#1\rangle}\nolimits}

\newcommand{\ave}[1]{\langle#1\rangle}

\def\bbC{\mathbb{C}}
\def\bbZ{\mathbb{Z}}
\def\bbR{\mathbb{R}}
\def\bbN{\mathbb{N}}

\def\a{\alpha}

\def\d{\delta}

\def\la{\lambda}

\sodef\so{}{.065em}{.4em plus1em}{2em plus.1em minus.1em}
\sodef\sosmalltitle{}{.065em}{.4em plus1em}{2em plus.1em minus.1em}

\begin{document}

\title[\sosmalltitle{Hiking a generalized Dyck path}]{\Large\so{Hiking a generalized Dyck path: A tractable way of calculating multimode boson evolution operators}}

\author{Kamil Br\'adler}
\email{kbradler\at ap.smu.ca}

\address{
    Department of Astronomy and Physics,
    Saint Mary's University,\\
    Halifax, Nova Scotia, B3H 3C3, Canada
    }


\begin{abstract}
  A time evolution operator in the interaction picture is given by exponentiating an interaction Hamiltonian $H$. Important examples of Hamiltonians, often encountered in quantum optics, condensed matter and high energy physics, are of a general form $H=r(\Adg-A)$, where $A$ is a multimode boson operator and $r$ is the coupling constant. If no simple factorization formula for the evolution operator exists, the calculation of the evolution operator is a notoriously difficult problem. In this case the only available option may be to Taylor expand the operator in $r$ and act  on a state of  interest $\psi$. But this brute-force method quickly hits the complexity barrier since the number of evaluated expressions increases exponentially. We relate a combinatorial structure called Dyck paths to the action of a boson word (monomial) and a large class of monomial sums on a quantum state~$\psi$. This allows us to cross the exponential gap and make the problem of a boson unitary operator evaluation computationally tractable by achieving polynomial-time complexity for an extensive family of physically interesting multimode Hamiltonians. We further test our method on a cubic boson Hamiltonian whose Taylor series is known to diverge for all nonzero values of the coupling constant and an analytic continuation via a Pad\'e approximant must be performed.
\end{abstract}

\maketitle

\section{Introduction}\label{sec:intro}

\thispagestyle{empty}

The calculation of an evolution unitary operator $V=\exp{[-iH]}$, where $H$ is an interaction Hamiltonian in the interaction picture,   is a generic problem encountered in quantum field theory, quantum optics or condensed matter physics, to name a few. Of considerable interest are single and mainly multimode boson Hamiltonians  describing the interaction between two and more boson field modes. The Hamiltonian $H$ is a Hermitian operator and can be written as $H=\nu\Adg+\overline\nu A$, where $\nu\in\bbC$ and the bar denotes complex conjugation. The operator $A$ is a function of one or more boson creation and annihilation operators $a_i$ and $a^\dg_i$, respectively, satisfying the Weyl-Heisenberg algebra $[a_i,a^\dg_j]=\d_{ij}\id$. The index $i$ denotes the $i$-th boson mode and $\id$ is an identity operator (not necessarily represented by a finite-dimensional identity matrix) commuting with $a_i$ and $a^\dg_i$. Because the Hamiltonian is often a sum of two or more non-commuting terms, the problem of finding an explicit form of $V$ is not a straightforward  one. The standard way of solving it is to use perturbation theory by expanding $V$ as the Dyson series~\cite{duncan2012conceptual,louisell1973quantum}. It can also can be converted to the problem of factorization of the operator exponential function known as the BCH formula (after Baker-Campbell-Hausdorff)~\cite{hall2003lie}. Strictly speaking, it is the Zassenhaus formula~\cite{magnus1954exponential,blanes2009magnus} that provides an explicit way of factorizing operator or matrix exponentials but it is not  the only, see~\cite{trotter1959product,wei1963lie,fer1958resolution,suzuki1976generalized,wilcox1967exponential,friedrichs1953mathematical}.  This extensive topic is also reviewed in~\cite{Mielnik1970} and more recently in~\cite{blanes2009magnus,bonfiglioli2011topics}.

By setting $\nu=s+ir$ we will restrict our attention to the evolution operator
\begin{equation}\label{eq:generalEvolutionOp}
  V=\exp{\big[r(A^\dg-A\big)\big]}.
\end{equation}
From all the existing approaches to evaluate Eq.~(\ref{eq:generalEvolutionOp}), the Taylor series expansion  around $r=0$ happens to be (most likely) the least sophisticated attempt to compute the exponential. But if no \emph{simple} factorization theorem exists (for example when $A$ and $\Adg$ commute) and other methods fail or are equally inefficient, it may be the only available option. As a matter of fact, short-time evolution is often studied by calculating the first few  terms of a Taylor expansion. However, for longer times the exponential growth in $k$ of the number of summands $\big(A^\dg\pm A\big)^k$ quickly makes such a calculation  intractable (note that we do not claim that all methods to evaluate $V$ are exponential in time).

The situation could have been saved by the following procedure. The expression $\big(A^\dg\pm A\big)^k$ can be brought to \emph{normal} form, where all creation operators $a^\dg_i$, forming $A$ and $\Adg$, are to the left of all annihilation operators $a_i$. We will denote the corresponding operation by $\Ncal$. The main advantage of normally ordered operators lies in an easy calculation of their action on many  states of interest, in particular, the eigenvector of~$a$ (a coherent state $\ket{\a}$ with the Minkowski vacuum as its special case). However, in general it is not easy to find  the normal form. As an example, take a boson word (also called product, string or monomial) $w=a^2\adg^3a$. The action of $\Ncal$ results in
\begin{equation}\label{eq:orderedExample}
\Ncal[w]=\adg^3a^3+6\adg^2a^2+6\adg a.
\end{equation}
A straightforward, but not recommendable, way to obtain~(\ref{eq:orderedExample}) is to commute through the forest of operators. Eventually, one can use shortcuts in the form of  derived commutation rules to make the process less painful~\cite{louisell1973quantum}. A better way is to take advantage of a clever differential representation of the aforementioned algebra~\cite{louisell1973quantum}
\begin{equation}\label{eq:HeisRep}
  a\leftrightarrow\a+\papa{}{\overline\a},\quad\adg\leftrightarrow\overline\a,
\end{equation}
where the bar denotes complex conjugation. Then the above calculation can be done in a blink of an eye (the celebrated Fock-Bargmann representation is another option~\cite{perelomov1986generalized}) but this is the exception rather than the rule. In the case of $\Ncal\big[\big(A^\dg\pm A\big)^k\big]$, the commutation relations, the differential operator representation~(\ref{eq:HeisRep}) or any other technique, such as Wick's theorem~\cite{duncan2012conceptual}, would lead to the right answer for but it may be difficult to obtain a general form for an arbitrary $A$ and $k$ or the computational complexity of the $\Ncal$ operator can be insurmountable.

The problem of boson  normal ordering has a long history and a number of techniques were developed~\cite{cahill1969ordered,katriel73,blasiak2005boson,blasiak2005combinatorics,mikhailov1983ordering,Witschel75,mansour2012combinatorics,varvak2005rook,Asakly,mansour2011commutation}. Certain special cases of $\Ncal\big[\big(A^\dg\pm A\big)^k\big]$ were previously studied  resulting in elegant combinatorial formulas.  But the general treatment for multimode boson operators is missing and the calculation of evolution operators generated by such Hamiltonians is the main goal of this paper. 

The idea behind our technique is to study $\big(A^\dg\pm A\big)^k$, and consequently the  expansion of Eq.~(\ref{eq:generalEvolutionOp}), in the weak sense, that is, by acting on a state of interest. After all, this is the situation most often encountered when solving physical problems -- the state is usually a vacuum state $\ket{\mathrm{vac}}$ (or a ground state defined  in general by $A\ket{\mathrm{vac}}=0$) and one is interested either in the vacuum expectation value of an operator or a unitary evolution of the  state it acts on. So our method lacks the generality of the factorization theorems for Eq.~(\ref{eq:generalEvolutionOp}) as a standalone operator~\cite{bonfiglioli2011topics,hall2003lie,blanes2009magnus} or of the  formulas developed to normal order certain special cases of  $\big(\Adg\pm A\big)^k$. However, the virtue of our technique is that it enables us to analytically  calculate
\begin{equation}\label{eq:generalOpSumWeak}
  \big(A^\dg\pm A\big)^k\psi^{(0)}
\end{equation}
for an \emph{arbitrary} multimode boson monomial $A$, where $\psi^{(0)}=\ket{\mathrm{vac}}$ defined above as $A\ket{\mathrm{vac}}=0$ or any state from the semi-infinite or finite tower of states generated by the repeated action of $\Adg$ on $\ket{\mathrm{vac}}$. In this work we will focus on the vacuum state as it the most often encountered  scenario. Equally importantly, the calculation will be shown to be tractable since it can be executed with polynomial-time complexity $\Ocal(dk^3)$ if $\Adg$ or $A$ is a monomial of the length $d$. The method does not depend on a particular algebraic structure of $A$ apart from the one inherited from $a,\adg$ satisfying the Weyl-Heisenberg algebra.  Our approach is combinatorial.  We will show that there exists a common structure called a generalized Dyck path, which is a subject of study in analytic combinatorics, computer science and stochastic processes, among others~\cite{stanley1999enumerative,mansour2008normal,flajolet2009analytic,mohanty1979lattice}. This insight will lead us directly to a recursive summing formula allowing to analytically calculate Eq.~(\ref{eq:generalOpSumWeak}) without the need for  normal ordering. The most important consequence is that it can be used to efficiently find (in polynomial time)  an analytical expression for the expanded evolution operator~(\ref{eq:generalEvolutionOp})
$$
V\psi^{(0)}=\sum_{k=0}^K{r^k\over k!}(A^\dg-A\big)^k\psi^{(0)},
$$
for any choice of $0<K<\infty$ and $0\leq r<\infty$ where again $\psi^{(0)}=\ket{\mathrm{vac}}$ or a tower of states generated by $\Adg$ acting on $\ket{\mathrm{vac}}$.

Multiboson Hamiltonians studied in this paper were often identified as \emph{quasi-exactly solvable}~\cite{alvarez2002quasi,bogoliubov1996exact} with a close relation to the boson realizations of polynomial algebras~\cite{links2003algebraic,lee2010polynomial}. Among the applicable physical scenarios, Bose-Einstein condensates and nonlinear optical systems are the most prominent. Quasi-exactly solvable models have been extensively studied~\cite{shift1989,gonzalez1994quasi,ushveridze1994quasi} and the difference compared to exactly solvable models is that the spectrum problem of such a Hamiltonian cannot be formulated as the usual eigenvalue problem in matrix algebra. However, there exists an invariant subspace of functions that are mapped to themselves and so it can be ``decoupled'' from the rest of an otherwise infinite-dimensional matrix Hamiltonian. This finite subspace  can subsequently be analytically diagonalized and part of the spectrum is recovered.

Generalized Dyck paths are introduced in Sec.~\ref{sec:DyckCount} and some of their combinatorial properties are summarized. In Sec.~\ref{sec:DyckBosons} we link Dyck paths to the action of boson words, show that all monomials behave as ladder operators  and uncover a large class of sums of multimode boson monomials that behave as ladder operators as well. This enables us to derive our main result in Sec.~\ref{sec:DyckEval}: the evaluation of Eq.~(\ref{eq:generalOpSumWeak}) for monomials and their sums. In Sec.~\ref{sec:complex} we discuss the computational complexity of our approach. We illustrate the evaluation on several examples of boson operators~$A$ and compute Eq.~(\ref{eq:generalOpSumWeak}) for low values of $k$ (to be able to verify the results by hand) and for some very high values of $k$ (to show the speed-up offered by our method). We  also show that the fast evaluation of Eq.~(\ref{eq:generalOpSumWeak}) leads to an efficient calculation of a boson evolution operator for previously hard-to-reach values of the coupling constant~$r$ if calculated by a series expansion. Our final example is the evolution operator for the three-photon degenerate process whose Taylor expansion is known to diverge for any value of the coupling constant. Yet our approach is still useful for the calculation of an analytically continued unitary operator using the Pad\'e approximation proposed in~\cite{braunstein1987generalized}.

\section{Generalized Dyck paths and their combinatorics}%\label{sec:DyckPath}
\subsection{Counting Dyck paths}\label{sec:DyckCount}

Let's introduce an integer lattice path known as a Dyck path and some of its properties~\cite{stanley1999enumerative,mansour2008normal,flajolet2009analytic,mohanty1979lattice}. The  lattice for Dyck paths is the set of all  integer points $\bbN\times\bbN\subset\bbR^2$. A Dyck  path $\euD(k)$ is a path starting at $(0,0)$ and ending at $(k,0)$, such that the only allowed steps (also called segments or letters in this paper) are $U=(1,1)$ and $D=(1,-1)$, that is, going to the nearest integer point northeast or southeast. The constraint on a Dyck path is that it never falls below the $x$ axis, but it may touch it at any number of integer points between $0$~and~$k$. A more general Dyck path $\euD(k,\d)$ starts at $(0,0)$ but is allowed to end at any point $(k,\d)$  following the same constraint. Finally, what we will call a generalized Dyck path $\euD(k,\d_1,\d_2)$ is a lattice path starting at $(0,\d_1)$ and ending at $(k,\d_2)$, where $k,\d_1,\d_2\geq0$, and as before, the path is constrained to the positive quadrant including the $x$ axis. The value of $\d_1$ can be any non-negative integer but the value of $\d_2$ depends on $k$ and $\d_1$. Assuming $\d_1$ even, if $k$ is even (odd) then only even (odd) values of $\d_2$ are possible. For $\d_1$ odd, if $k$ is even (odd) then only odd (even) values of $\d_2$ are possible.

Two examples of a generalized Dyck path are in Fig.~\ref{fig:DyckExamples}. Note that $\euD(k,\d)\equiv\euD(k,0,\d)$ and from now on whenever we say a Dyck path we will mean a generalized Dyck path\footnote{Sometimes, a Dyck path is called generalized if two steps are allowed -- not necessarily being $U=(1,1)$ and $D=(1,-1)$~\cite{pemantle2013analytic,labelle1990generalized}. This is not the kind of generalization considered in this paper.}.
\begin{figure}[t]
    \centering
  \includegraphics{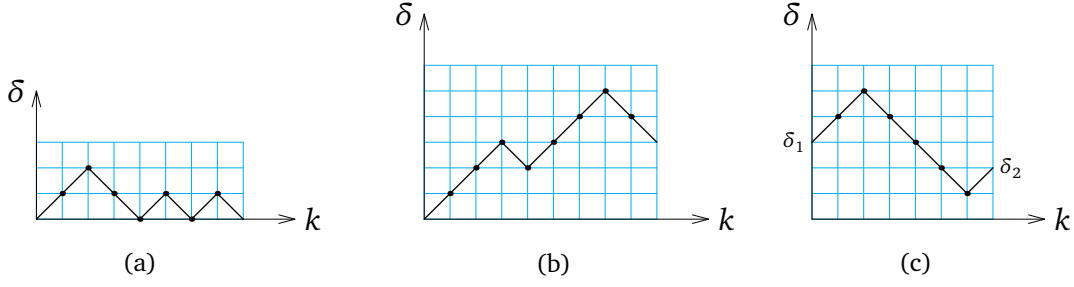}
   \caption{We plot (a) Dyck path $\euD(8,0,0)$ (b) more general Dyck path $\euD(10,0,3)$ and (c) the generalized Dyck path $\euD(7,3,2)$. Dyck paths are often described as mountain ranges.}
   \label{fig:DyckExamples}
\end{figure}
A Dyck path is unambiguously defined by its starting point $\d_1$ and a set of instructions called a Dyck word that leads the path through the lattice. It is a string of the segments $U$ and $D$ whose length is $k$ and our convention will be to read the string from the right. The highest Dyck path $\hiD(k,\d_1,\d_2)$ is defined by the (highest) Dyck word, where all the $U$ segments are on the right:
$$
\overline{W}=\underbrace{D\dots\dots D}_{k+\d_1-\d_2\over2}\underbrace{U\dots\dots U}_{k-\d_1+\d_2\over2}.
$$
Then, all Dyck paths for the fixed parameters $k,\d_1$ and $\d_2$  lie ``between'' the highest Dyck path and the $x$ axis. Equivalently, the lowest Dyck path $\loD(k,\d_1,\d_2)$ is defined such that there is no other Dyck path between this path and the $x$ axis for the chosen set of parameters. It is generated by the lowest Dyck word
$$
\underline{\smash W}=\underbrace{U\dots\dots U}_{\d_2}\underbrace{DU\dots\dots DU}_{k-\d_1-\d_2}\underbrace{D\dots\dots D}_{\d_1}.
$$
Any Dyck path with the given parameters $k,\d_1,\d_2$ contains the same number of $U$ and $D$ segments.

The number of Dyck paths  $\euD(k,0,0)$  is famously counted by the Catalan number $C_{k/2}$
\begin{equation}\label{eq:Catalan}
  |\euD(k,0,0)|\equiv C_{k/2}={2\over2+k}\binom{k}{{k\over2}}.
\end{equation}
Catalan numbers are ubiquitous in computer science and discrete mathematics~\cite{koshy2008catalan} and this makes Dyck paths  a well studied object appearing in many reincarnations~\cite{stanley1999enumerative}\footnote{A huge collection of examples is available  at \url{http://www-math.mit.edu/~rstan/ec} as an addendum to~\cite{stanley1999enumerative}.}. One of them is Bertrand's ballot problem and by virtue of this insight one finds~\cite{feller1967} the number of Dyck paths $\euD(k,\d_1,\d_2)$ to be
\begin{equation}\label{eq:mostGeneralDyckCounting}
  |\euD(k,\d_1,\d_2)|\equiv G(k,\d_1,\d_2)=\binom{k}{{1\over2}(k+\d_2-\d_1)}-\binom{k}{{1\over2}(k-\d_2-\d_1-2)},
\end{equation}
where the binomial coefficient is defined to be zero whenever the ``denominator'' is negative.
\begin{figure}[t]
\centering
  \includegraphics{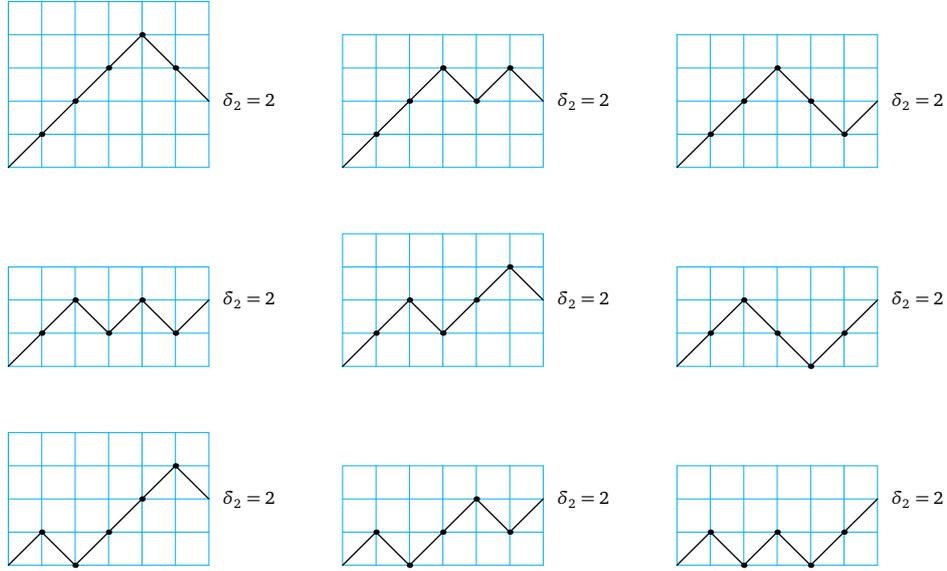}
   \caption{All $G(6,0,2)=9$ Dyck paths of length $k=6$ starting at $\d_1=0$ and ending at $\d_2=2$ are depicted.}
   \label{fig:AllDycks}
\end{figure}

As an illustration we plotted all Dyck paths $\euD(6,0,2)$ in Fig.~\ref{fig:AllDycks} for which the expression from Eq.~(\ref{eq:mostGeneralDyckCounting}) simplifies to
\begin{equation}\label{eq:moreGeneralDyckCounting}
  G(k,0,\d_2)=\frac{\d_2+1 }{\frac{\d_2+k}{2}+1}\binom{k}{\frac{\d_2+k}{2}}.
\end{equation}
Note that if $k=\d_1+\d_2$ in (\ref{eq:mostGeneralDyckCounting}), all the Dyck paths $\euD(k,k-\d_2,\d_2)$ are confined to a rectangle whose two corners are touching the positive $x$ and $y$ axes and
$$
G(k,k-\d_2,\d_2)=\binom{k}{\d_2}
$$
is a binomial number. This is nothing else than a binomial walk on the Pascal triangle (rectangle) experimentally realized by the Galton board.

\subsection{Dyck paths and bosons}
\label{sec:DyckBosons}

Let $A$  be a general multimode boson operator such that there exists $\psi$ satisfying~$A\ket{\mathrm{vac}}=0$. So we may call $\ket{\mathrm{vac}}$ a vacuum or ground  state. The object of our study is the expression  $\Adg\pm A$ that can be considered as a hypothetical Hamiltonian\footnote{The minus version needed to evaluate~(\ref{eq:generalEvolutionOp}) is recovered by putting a minus to all products with an odd number of $A$'s. The appearance of the plus sign is irrelevant for the presented method.}. The operator $\Adg$ acts  repeatedly on the ground state and generates a finite or eventually semi-infinite tower of states: $\Adg^p\ket{\mathrm{vac}}=\widetilde\psi^{(p)}$, where $\ket{\mathrm{vac}}\equiv\psi^{(0)}$ and $\widetilde\psi^{(p)}$ is, in general, an unnormalized state. The first crucial question is under what conditions the adjoint operator $A$  behaves as  an annihilation operator:
\begin{subequations}\label{eq:LadderOps}
\begin{align}%\label{}
    \Adg\psi^{(p-1)}&=\la_p\psi^{(p)},\\
    A\psi^{(p)}&=\mu_p\psi^{(p-1)}.
\end{align}
\end{subequations}
We adopted the convention that the states with no tildes are normalized and so the coefficients $\la_p$ and $\mu_{p}$  are determined by this condition. The second crucial point we have to address in order for our method to work is that our approach hinges on the ability to analytically sum over the product $\la_{p}\mu_{p}$ later derived  in Sec.~\ref{sec:DyckEval}  and used in Eqs.~(\ref{eq:mainSum1}) and~(\ref{eq:mainSum2}) (appearing as $\la_{m_i}\mu_{m_i}$). First, we will show that $\la_p\mu_p$ can be obtained even without knowing the normalization of $\psi^{(p)}$ (note that we are not interested in $\la_p$ and $\mu_p$ but only in their product). Finding the normalization may not only be rather laborious  for a generic boson expression $A$ and all $p$ but we even cannot guarantee the existence of a closed, or even simple, form. So assume for a moment that $A,\Adg$ behave as proper ladder operators as in (\ref{eq:LadderOps}). Then we may write
\begin{align}%\label{}
    A\,\Adg\psi^{(p-1)}&=\la_pA\psi^{(p)}=\la_p\mu_p\psi^{(p-1)}.
\end{align}
But this clearly does not depend on whether $\psi^{(p-1)}$ is normalized or not. It holds that $\psi^{(p-1)}=\mathsf{N}_{p-1}\widetilde\psi^{(p-1)}$, where $\Adg^{p-1}\ket{\mathrm{vac}}=\widetilde\psi^{(p-1)}$ and $\mathsf{N}_{p-1}$ is the normalization of $\widetilde\psi^{(p-1)}$. Hence we get the sought product $\la_{p}\mu_{p}$ from the unnormalized ladder states for all $p$
\begin{align}\label{eq:LadderOpsUnnormalized}
    A\,\Adg\widetilde\psi^{(p-1)}&=\la_p\mu_p\widetilde\psi^{(p-1)}
\end{align}
greatly simplifying its derivation.

Let's now uncover a large class of the boson expressions $A,\Adg$ that behaves as ladder operators. Remarkably, we will show that this is equivalent to the condition that  the product $\la_{p}\mu_{p}$ is an integer and therefore leads to the announced speed-up derived in  Sec.~\ref{sec:DyckEval}.
\begin{defi}
  A normally ordered monomial of the length $d$ is the expression
  \begin{equation}\label{eq:monomial}
    M^{(\Ncal)}={a_0^\dg}^{k_0}{a_0^{\ell_0}}{a_1^\dg}^{k_1}{a_1^{\ell_1}}\hdots {a_d^\dg}^{k_d}{a_d^{\ell_d}}.
  \end{equation}
  We define a normal \emph{disjoint monomial of the same order} as any normally ordered monomial that is created by taking the creation and annihilation exponents $(k_0,k_1,\dots,k_d)$ and  $(\ell_0,\ell_1,\dots,\ell_d)$ and (independently) permuting them such that no mode with a nonzero exponent in $M^{(\Ncal)}$ is present in the newly created one.
\end{defi}
The definition seems restrictive but the opposite is true. It allows for a large number of multimode boson expressions that we will appreciate in the example after the following statement.
\begin{prop}
Let $A$ be a sum of disjoint multimode boson monomials of the same order and $\ket{\mathrm{vac}}$ a vacuum state  such that $A\ket{\mathrm{vac}}=0$. Then the operator $A\,\Adg$ is proportional to an identity map for all ladder states $\widetilde\psi^{(p-1)}\df \Adg^{p-1}\ket{\mathrm{vac}}$, that is, $A\,\Adg\widetilde\psi^{(p-1)}=q_{p-1}\widetilde\psi^{(p-1)}$. Moreover, the constant of proportionality $q_{p-1}$ is an integer.
\end{prop}
\begin{rem}
   Note that a vacuum (or ground state) is in general not unique. A given operator $A$ can annihilate several distinct vacui.
\end{rem}
\begin{exa}
   The boson expressions $A$ can now be constructed in a great variety. The trivial example is every monomial $M$ (arbitrarily ordered). To illustrate the definition we further restrict to normal monomials $M^{(\Ncal)}$ but later we will comment on differently ordered alternatives. If we fix the monomial order $d$ then the first nontrivial example is $M^{(\Ncal)}={a_0^\dg}^{k_0}$ for any $k_0$. The disjoint monomials of the same order are all $M^{(\Ncal)}={a_i^\dg}^{k_0}$, where $1\leq i\leq d$ and so the boson expression $A$ can be any sum
    \begin{equation}\label{eq:disjointsum1}
       A=\sum_{i=0}^{\leq d}{a_i^\dg}^{k_0}.
    \end{equation}
   Another example is  $M^{(\Ncal)}=a^\dg_0a_0{a_1^\dg}^{3}$. A sum of disjoint multimode boson monomials of the same order can be, for instance,
    \begin{equation}\label{eq:disjointsum2}
       A=a^\dg_0a_0{a_1^\dg}^{3}+a^\dg_2a_2{a_3^\dg}^{3}+a^\dg_4a_4{a_5^\dg}^{3}.
    \end{equation}
\end{exa}
To prove the proposition we will need an auxiliary lemma.
\begin{lem}
  Let $M_1$ and $M_2$ be  boson multimode monomials such that $M_2$ is an arbitrarily reordered version of $M_1$. Then
    \begin{equation}\label{eq:M1M2}
      M_1=M_2+f(n_i),
    \end{equation}
  where $f(n_i)$ is a multivariate polynomial with the variables being the $i$-th mode boson numbers $n_i=a_i^\dg a_i$ for $0\leq i\leq d$.
\end{lem}
\begin{proof}
   Any monomial $M_1$ (normal ordered or not) can always be recast to the form of $M_2$ resulting in Eq.~(\ref{eq:M1M2}). Then, all summands of the (non-unique) polynomial $f(n_i)$ must contain an equal number  of creation and annihilation operators for all modes involved in $M_1$ and $M_2$. If it was not the case then the lemma assumption about $M_1$ and $M_2$ would be violated. The polynomial $f(n_i)$ can be put (by normal ordering) to a form  containing only the expressions ${a_i^\dg}^pa_i^p$  and every such building block can be rewritten as a sum of powers of $n_i$.
\end{proof}
\begin{proof}[Proof of proposition]
  We will be concerned with the sums of normally ordered disjoint multimode boson monomials denoted $A$ and $\Adg$ and their corresponding products $A\,\Adg$ and $\Adg A$. Each monomial in $A\,\Adg$ has a reordered counterpart with the same number of creation and annihilation operators in $\Adg A$. According to the previous lemma, this implies
    \begin{equation}\label{eq:AAdg}
      A\,\Adg=\Adg A+g(n_i),
    \end{equation}
  where $g(n_i)$ is a multivariate polynomial with the variables being the boson number operators $n_i$. Furthermore, using the distributivity of the commutator %$[a+b,c+d]=[a,c]+[a,d]+[b,c]+[b,d]$
  and the crucial fact that
  $$
  A=\sum_IA_I
  $$
  and its adjoint $\Adg$ are disjoint ($I$ is a disjoint multi-index, see the examples in Eqs.~(\ref{eq:disjointsum1}) and (\ref{eq:disjointsum2})), we may write
  $$
  g(n_i)=\sum_If(n_I),
  $$
  where $[A_I,\Adg_I]=f(n_I)$. This is the only surviving term in commutator~Eq.~(\ref{eq:AAdg}). To use this fact, we write
\begin{subequations}\label{eq:commutantis1}
    \begin{align}%\label{eq:}
       A\,\Adg\psi^{(p-1)}&=A\,\Adg\,{{A^\dg}}^{p-1}\ket{\mathrm{vac}} \\
       & = \left(\Adg A+g(n_i)\right){{A^\dg}}^{p-1}\ket{\mathrm{vac}}\\
       & = \left(\Adg A\,{{A^\dg}}^{p-1}+{{A^\dg}}^{p-1}g^{(p-1)}(n_i)\right)\ket{\mathrm{vac}},
    \end{align}
\end{subequations}
where in the second summand of the third row we used the following consequence of the boson commutation rule:
\begin{align}\label{eq:ncommuting}
  n_i{a_j^\dg}^p & = {a_j^\dg}^p(n_i+\d_{ij}p), \\
  n_ia_j^p & = a_j^p(n_i-\d_{ij}p).
\end{align}
So we ``propagated'' the polynomial $g(n_i)$ to the right side of ${A^\dg}^{p-1}$ converting it to another polynomial $g^{(p-1)}(n_i)$, where only the boson number operators appear as variables. Recall that $\Adg$ is a sum of disjoint monomials and so all but one $f(n_I)$ commutes and the one that does not commute switches the side according to the rules~Eqs.~(\ref{eq:ncommuting}) in the same way for all $f(n_I)$. This results in $g(n_i)\Adg=\Adg g^{(1)}(n_i)$ implying $g(n_i){A^\dg}^{p-1}={A^\dg}^{p-1}g^{(p-1)}(n_i)$.  We continue by writing
\begin{subequations}\label{eq:commutantis2}
    \begin{align}
      \Adg A\,{A^\dg}^{p-1} & = \Adg \left(\Adg A+g(n_i)\right) {A^\dg}^{p-2}\\
      & = {A^\dg}^{2}\,A\,{A^\dg}^{p-2}+{A^\dg}^{p-1}g^{(p-2)}(n_i)
    \end{align}
\end{subequations}
until we arrive at
\begin{equation}\label{eq:final}
  A\,\Adg\widetilde\psi^{(p-1)}=\Big[{A^\dg}^{p}A+{A^\dg}^{p-1}\sum_{s=1}^{p-1}g^{(p-s)}(n_i)\Big]\ket{\mathrm{vac}}
  =q_{p-1}\widetilde\psi^{(p-1)},
\end{equation}
where we used $A\ket{\mathrm{vac}}=0\Rightarrow{A^\dg}^{p}A\ket{\mathrm{vac}}=0$ and the fact that  only the number coefficients of $g^{(p-s)}(n_i)$ survive the encounter with $\ket{\mathrm{vac}}$ generating $q_p\in\bbZ$.
\end{proof}
\begin{rem}
  The $A$ expressions are normally ordered but we might  have equally used antinormally ordered monomials~\cite{louisell1973quantum,cahill1969ordered} or arbitrarily ordered ones and prove the proposition for them. It would lead to an ever larger class of operators satisfying the proposition.
\end{rem}
\begin{rem}
  The class of operators uncovered in the proposition is certainly not the most general class of boson expressions with this kind of ladder-like behavior. Yet, among this class, the monomials represent a huge variety of hypothetical Hamiltonians and so we will mainly focus on them.
\end{rem}
\begin{exa}\label{exa:splitter}
    As our first example, consider $\Adg=a_0^\dg a_1$. The ground state is $\psi^{(0)}=\ket{0}_0\ket{n}_{1}$ since $A\psi^{(0)}=0$ for all $n$ and  $\Adg$ (acting repeatedly on $\psi^{(0)}$) generates a finite tower of states $\psi^{(p)}$ for $0\leq p\leq n$. The corresponding Hamiltonian $H_{BS}=r(a_0^\dg a_1-a_0a_1^\dg)$ is incidentally an important one. It represents one of the generators of the $su(2)$ algebra  describing the action of a quantum-optical beam-splitter and the fact that the ground state is somewhat non-unique turns out to be important too. The $l=n+1$ states generated by $\Adg^p\ket{0}_0\ket{n}_{1}$ carry the $l$-th representation of $su(2)$~\cite{yurke1986}. The claims from the theorem are manifestly satisfied following the standard properties of the boson operators~\cite{louisell1973quantum}
    \begin{subequations}\label{eq:bosonSteps}
      \begin{align}%\label{}
        \adg^m\ket{n} & =\sqrt{(n+m)!\over n!}\ket{n+m}, \\
        a^m\ket{n}   & =\sqrt{n!\over (n-m)!}\ket{n-m}.
      \end{align}
    \end{subequations}
\end{exa}
\begin{exa}\label{exa:contrived}
  The next example will be $A$ in the form of the following disjoint sum:
\begin{equation}\label{eq:sumA}
  A=\aadg{0}a_2a_3a_4+\aadg{1}a_5a_6a_7.
\end{equation}
Assume the ground state of $A$ to be $\psi^{(0)}=\ket{n}_0\ket{n}_{1}\ket{0}_{2-7}$. We find
$$
\Adg\psi^{(0)}=\sqrt{n}
\big(\ket{n-1}_0\ket{n}_{1}\ket{1}^{\otimes3}_{234}\ket{0}_{567}^{\otimes3}+\ket{n}_0\ket{n-1}_{1}\ket{0}_{234}^{\otimes3}\ket{1}_{567}^{\otimes3}\big)
=\sqrt{2n}\psi^{(1)}=\la_1\psi^{(1)}.
$$
We can go on and derive a general expression for $\la_p$ and $\mu_p$.
\end{exa}
Since the state on the right of the previous equation is normalized, the $\la_1$ contains a square root and $\la_p$ could be quite difficult to find (and unnecessary since we need $\la_p\mu_p$ in the end). So in the next, morally similar, example we derive it using our insight from Eq.~(\ref{eq:LadderOpsUnnormalized}).
\begin{exa}\label{exa:Calculable}
  Let
  $$
  A=\aadg{0}a_2+\aadg{1}a_3
  $$
  and $\psi^{(0)}=\ket{n}_0\ket{n}_{1}\ket{0}_{2}\ket{0}_3$. We find
    \begin{align}%\label{}
      \Adg^{p-1}\ket{\mathrm{vac}}&=\widetilde\psi^{(p-1)}\nn \\
       & =(p-1)!\sum_{m=0}^{p-1}\sqrt{\binom{n}{m}}\sqrt{\binom{n}{p-1-m}}\ket{n-m}_0\ket{n-p+1+m}_{1}\ket{m}_{2}\ket{p-1-m}_3.
    \end{align}
  By matching the coefficients of the corresponding basis vectors, the calculation
  $$
  A\,\Adg\widetilde\psi^{(p-1)}=\la_p\mu_p\widetilde\psi^{(p-1)}=p(2n-p+1)\widetilde\psi^{(p-1)}
  $$
  reveals the constant of proportionality. We can verify it by a direct calculation for $p=2$:
    \begin{align}%\label{}
      A\,\Adg\widetilde\psi^{(1)}&=(\aadg{0}a_2+\aadg{1}a_3)(a_0\aadg{2}+a_1\aadg{3})
      \big(\sqrt{n}\ket{n}_0\ket{n-1}_{1}\ket{0}_{2}\ket{1}_3+\sqrt{n}\ket{n-1}_0\ket{n}_{1}\ket{1}_{2}\ket{0}_3\big)\nn\\
       & = (4n-2)\big(\sqrt{n}\ket{n}_0\ket{n-1}_{1}\ket{0}_{2}\ket{1}_3+\sqrt{n}\ket{n-1}_0\ket{n}_{1}\ket{1}_{2}\ket{0}_3\big)\nn\\
       & \equiv p(2n-p+1)\big|_{p=2}\psi^{(1)}.\nn
    \end{align}
\end{exa}

Convinced about the validity of the proposition we can set out to show its main application. The sum $(\Adg+A)^k$ contains $2^k$ summands but when acting on $\psi$, some of them become zero. What is the exact condition for a general summand of $(\Adg+A)^k$ to survive its encounter with the ground state? To find out, we write
\begin{equation}\label{eq:Expsummand}
  (\Adg+A)^k=\sum_{m=0}^k\sum^{\binom{k}{m}}_{|S|=1}(A)^{m_1}(\Adg)^{m_2}\dots(A)^{m_{2J-1}}(\Adg)^{m_{2J}},
\end{equation}
where
\begin{equation}\label{eq:set}
S\df\Bigg\{m_1,\dots,m_{2J};m_j\geq0,\sum_{j=1}^{2J}m_j=k,\sum_{j=1}^Jm_{2j}=m\Bigg\}.
\end{equation}
At least one half of all the products on the RHS of Eq.~(\ref{eq:Expsummand}) vanishes upon acting on $\psi$ because the number of $A$ operators outnumber their conjugates. This is a necessary but not sufficient condition to eliminate all the products that do not contribute. For example, if there is just one $A$ operator that is, however, the first one to encounter $\psi$, the whole summand disappears. A general product in Eq.~(\ref{eq:Expsummand}) is a bosonic word of the length $k$ according to $(\ref{eq:set})$
\begin{equation}\label{eq:bosonicWord}
(A)^{m_1}(\Adg)^{m_2}\dots(A)^{m_{2J-1}}(\Adg)^{m_{2J}}.
\end{equation}
From the way it acts on $\psi$ it can be seen that the word yields zero if and only if the number $m_{2i-1}$ of annihilation  operators $A$  on the $(2i-1)$-st position is greater than the number $\sum_{j=i}^{J} m_{2j}$ of {\em all}  creation operators  $\Adg$ to its right valid for {\em any} $1\leq i\leq J$. But this is precisely the defining property of a Dyck path $\euD(k,0,\d_2)$. We associate the step operators with the lattice steps
\begin{subequations}\label{eq:bijection}
\begin{align}%\label{}
  A & \leftrightharpoons D=(1,-1), \\
  \Adg & \leftrightharpoons U=(1,1)
\end{align}
\end{subequations}
and a bosonic word is nothing else than a Dyck word of instructions of how a Dyck path is generated. Indeed, the constraint on a Dyck path to remain confined in the positive quadrant is equivalent to the condition for a bosonic word (product) in Eq.~(\ref{eq:Expsummand}) to be nonzero. The converse is true as well because of the way the annihilation operator acts on a ground state and so we have a bijection between the set of  Dyck paths $\euD(k,0,\d_2)$ of cardinality $G(k,0,\d_2)$ and the number of products on the RHS of
\begin{equation}\label{eq:ExpFinalExpansion}
(\Adg+A)^k\psi=\sum_{m=\lfloor k/2\rfloor}^k\sum_{|S|=1}^{G(k,0,\d_2)}(A)^{m_1}(\Adg)^{m_2}\dots(A)^{m_{2J-1}}(\Adg)^{m_{2J}}\psi,
\end{equation}
where
\begin{subequations}%\label{}
  \begin{align}%\label{}
    k & = \sum_{j=1}^Jm_{2j}+\sum_{j=1}^Jm_{2j-1},\\
    \d_2 & = \sum_{j=1}^Jm_{2j}-\sum_{j=1}^Jm_{2j-1}.
  \end{align}
\end{subequations}
The first line is extracted from the definition of $S$ (the length of the boson word) and the second line is the difference between the number of $A$ and $\Adg$ determining the $y$ component of the end point of the corresponding Dyck path. Note  the lower limit of the outer sum in Eq.~(\ref{eq:ExpFinalExpansion}). It removes the words with a  number of annihilation operators greater than the number of creation operators. They do not correspond to any Dyck path.

We may associate the $p$-th horizontal line of the integer lattice with the state  $\psi^{(p)}$. The $x$ axis is then the ground state $\psi$. With this insight in hand, we generalize the above construction for an arbitrary starting state $\psi^{(\d_1)}$, where we set $\d_1=p$:
\begin{equation}\label{eq:ExpFinalExpansionGeneral}
(\Adg+A)^k\psi^{(\d_1)}=\sum_{m=\lfloor k/2\rfloor-\lfloor\d_1\rfloor}^k\sum_{|S|=1}^{G(k,\d_1,\d_2)}(A)^{m_1}(\Adg)^{m_2}\dots(A)^{m_{2J-1}}(\Adg)^{m_{2J}}\psi^{(\d_1)}.
\end{equation}
Note that $\d_1$ can be greater than $k$ in which case the lower limit of the outer sum is zero.

To illustrate our technique, we will study in this paper the ground state case $\d_1=0$ corresponding to Eq.~(\ref{eq:ExpFinalExpansion}). The analysis for an arbitrary ladder state $\psi^{(\d_1)}$ can be easily generalized following the approach developed in the next section.

\subsection{Evaluating Dyck paths}
\label{sec:DyckEval}

In this section we present our main result. For that purpose we introduce a different kind of Dyck path description equivalent to a Dyck word~$W$.  Previously we introduced the highest and lowest Dyck paths, $\hiD(k,\d_1,\d_2)$ and $\loD(k,\d_1,\d_2)$, respectively. They are illustrated in~Fig.~\ref{fig:maxminDyck} for $k=14,\d_1=0$ and $\d_2=0,2$. The letter $u_i$ denotes the {\em ascending} path segment lying between the $(i-1)$-st and $i$-th horizontal line of the lattice. Similarly, $d_i$ denotes the {\em descending} segment connecting the $i$-th and $(i-1)$-st horizontal lattice line. At this point, we treat $u_i,d_i$ as letters of a different (somewhat less economical) kind of Dyck word we will label~$w$  which is, nonetheless, equivalent to the usual Dyck word $W$. A Dyck word $w$ is, like $W$, read from the right and we write down the highest Dyck path $\overline{W}$ from Fig.~\ref{fig:maxminDyck}~(a) as an example:
\begin{equation}\label{eq:DyckWords}
  \overline{W}=D^7U^7\equiv d_1d_2d_3d_4d_5d_6d_7u_7u_6u_5u_4u_3u_2u_1=\overline{w}.
\end{equation}
\begin{figure}[t]
    \centering
  \includegraphics{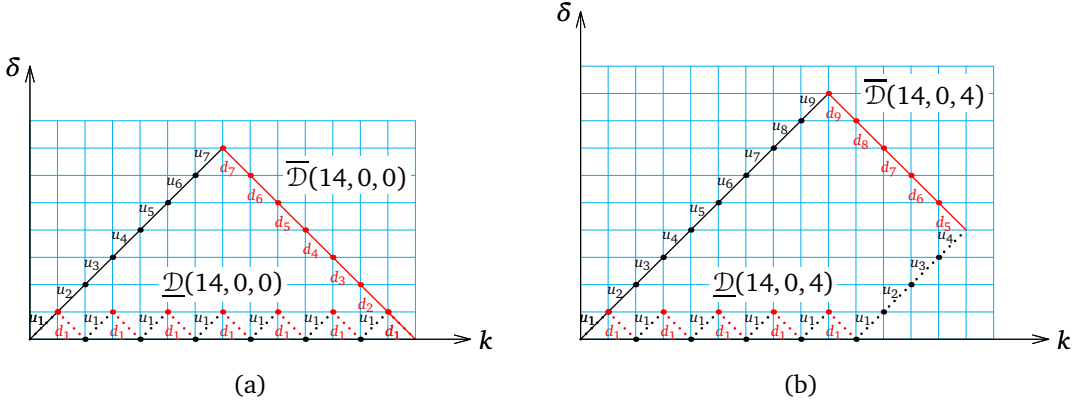}
   \caption{Two examples of the highest (solid)  and lowest (dotted) Dyck path for a given $k,\d_2$ and $\d_1=0$. Two different colors distinguish the ascending (black) and descending (red) sections of the paths. The  letters $u_i$ and $d_i$ form a Dyck word and their role is explained in Sec.~\ref{sec:DyckEval}.}
   \label{fig:maxminDyck}
\end{figure}
Dyck paths are continuous and so two consecutive segments of {\em any} Dyck path must be one of these four pairs (recall the right-to-left convention):
\begin{equation}\label{eq:GeneralDyckSegments}
\begin{aligned}
  & u_{i+1}u_i, \\
  & d_iu_i, \\
  & u_id_i, \\
  & d_id_{i+1}.
\end{aligned}
\end{equation}
It follows that if we have the highest Dyck path, where all $u_j$ are on the right, and there is a need to swap the leftmost letter $u_i$  with the rightmost letter $d_i$ such that the new word represents a Dyck path, the general rule is very simple:
\begin{equation}\label{eq:TheRule}
  d_iu_i\to u_{i-1}d_{i-1},
\end{equation}
valid for all admissible $i$. Only then the new word satisfies the conditions in~(\ref{eq:GeneralDyckSegments}) no matter what letter precedes $d_i$ or follows $u_i$ on the LHS of~(\ref{eq:TheRule}). 

The rule can be immediately put to work  by introducing a  way of how to generate all Dyck paths $\euD(k,\d_1,\d_2)$. We ``descend'' from  the highest Dyck path $\hiD(k,\d_1,\d_2)$ to the lowest Dyck path $\loD(k,0,0)$ by systematically swapping the highest ascending and descending segments $u_i$ and $d_i$ and generating all Dyck paths ending at $(k,\d_2)$ such that no Dyck path is left behind.

The rule is rather straightforward so let's illustrate it on the example $k=14$ depicted in Fig.~\ref{fig:maxminDyck}~(a). The first step is to swap the leftmost $U$ letter of the highest Dyck word $\overline{W}=D^7U^7$ with all the $D$ letters to its left except for the leftmost one. This generates the following set of words (the swapped letters are in bold to emphasize the transformation):
\begin{equation}\label{eq:FirstIter}
\begin{aligned}
   & D^6\Ubi\Dbi U^6 \\
   & D^5\Ubi\Dbi\Dbi U^6 \\
   & D^4\Ubi\Dbi\Dbi\Dbi U^6 \\
   & D^3\Ubi\Dbi\Dbi\Dbi\Dbi U^6 \\
   & D^2\Ubi\Dbi\Dbi\Dbi\Dbi\Dbi U^6 \\
   & D\Ubi\Dbi\Dbi\Dbi\Dbi\Dbi\Dbi U^6.
\end{aligned}
\end{equation}
Looking at the generated Dyck path in Fig.~\ref{fig:maxminDyck}~(a) we observe that all Dyck paths reaching the level~$u_6$ were obtained. Following the swap rule Eq.~(\ref{eq:TheRule}), the same operation using the Dyck path representation given by the $w$ word reads
\begin{equation}\label{eq:FirstIter2}
\begin{aligned}
   d_1d_2d_3d_4d_5d_6\ubi_6\dbi_6&u_6u_5u_4u_3u_2u_1\\
   d_1d_2d_3d_4d_5\ubi_5\dbi_5\dbi_6&u_6u_5u_4u_3u_2u_1 \\
   d_1d_2d_3d_4\ubi_4\dbi_4\dbi_5\dbi_6&u_6u_5u_4u_3u_2u_1 \\
   d_1d_2d_3\ubi_3\dbi_3\dbi_4\dbi_5\dbi_6&u_6u_5u_4u_3u_2u_1 \\
   d_1d_2\ubi_2\dbi_2\dbi_3\dbi_4\dbi_5\dbi_6&u_6u_5u_4u_3u_2u_1 \\
   d_1\ubi_1\dbi_1\dbi_2\dbi_3\dbi_4\dbi_5\dbi_6&u_6u_5u_4u_3u_2u_1.
\end{aligned}
\end{equation}
Recall that in~(\ref{eq:FirstIter2}) we do exactly the same operation as in~(\ref{eq:FirstIter}), just with some more information attached to the letters of the Dyck word~$w$. A similar method can be used to generate all Dyck paths $\euD(k,0,\d_2)$ by starting from $\overline{\euD}(k,0,\d_2)$.

In the next step we would have swapped two $U$'s and let both gradually propagate to the last allowed word. But already in this case the resulting words get complicated and so we need to simplify our strategy. Before we do so it is perhaps time to reveal the reason why. In the crucial step that follows, we associate the letters $u_i$ and $d_i$ of all Dyck paths, generated according to the rule Eq.~(\ref{eq:TheRule}), with the scalar coefficients given by Eqs.~(\ref{eq:LadderOps})
\begin{subequations}\label{eq:symbol2real}
  \begin{align}
    u_i & \mapsto \la_i \label{eq:symbol2realA}\\
    d_i & \mapsto \mu_i\label{eq:symbol2realB}
  \end{align}
\end{subequations}
resulting in
\begin{equation}\label{eq:typeChange}
  w\mapsto\prod_i\la_i\prod_i\mu_i=x.
\end{equation}
The number $x$ is the result of a transformation we call \emph{evaluation} of a Dyck path $w$.

We already linked boson words Eq.~(\ref{eq:bosonicWord}) with Dyck words $W$ via bijection~(\ref{eq:bijection}). Because $w$ is equivalent to $W$, we automatically linked boson words with Dyck words $w$.  But here we go further thanks to the fact that  $w$ carries much more information. It allows not only to properly count the boson words but also to relate the result of their action on a vacuum state $\psi$ via the assignment in Eqs.~(\ref{eq:symbol2real}).

Now it is clear what our goal is: we collect all $w$'s, evaluate them and sum them. This will be equivalent to the action of a sum of all bosonic monomials with the total number of $k$ operators $A$ and $\Adg$ acting on $\psi$, whose difference between the number of creation and annihilation operators is $\d_2$ as a consequence of Eqs.~(\ref{eq:symbol2real}) and~(\ref{eq:bijection}). The operations~(\ref{eq:symbol2real}) and (\ref{eq:typeChange})  are peculiar since we effectively changed the variable's type from a symbol and word to a real number. The number $x$ is not, in general, unique since the scalars $\la_i,\mu_i$ commute and it may happen that two different Dyck words $w$ lead to the same $x$. But this does not bother us since $x$ will not be used to identify a Dyck path.

In order to execute the transformation Eq.~(\ref{eq:typeChange}) it is critical  to be able to list all Dyck paths. Here comes the power of Eq.~(\ref{eq:TheRule}). It informs us that the set of indices $i$ of every Dyck path $w$ is identical for the ascending ($u_i$) and descending ($d_i$) sections\footnote{Cf. the indices of every row in~(\ref{eq:FirstIter2}) -- the indices are the same for $u$'s and $d$'s but they are not ordered in a particularly lucid way.}. Hence to systematically list all Dyck paths $\euD(k,0,\d_2)$ as the offspring of the highest Dyck path we may only focus on the ascending or descending sections of $w$. But it is very easy to list  ascending segments only. Every Dyck path $\euD(k,0,\d_2)$  contains the following ascending segments: $\{u_i\}_{i=1}^{(k+\d_2)/2}$ (see the properties of Dyck paths discussed in Sec.~\ref{sec:DyckCount}). As follows from  the first row of~(\ref{eq:GeneralDyckSegments}), the difference of indices of two consecutive ascending segments must be one ($(i+1)-i=1$). This is the greatest possible difference because if two ascending segments $u_i$ and $u_j$ are separated by several descending segments $\{d_k\}_k$ it follows from the second and third line of~(\ref{eq:GeneralDyckSegments}) that $j-i<1$.
We just rephrased a simple fact easily seen by inspecting any Dyck path (see Figs.~\ref{fig:DyckExamples},~\ref{fig:AllDycks} and~\ref{fig:maxminDyck}). Now we know how to generate all ascending sections of a Dyck path $\euD(k,0,\d_2)$ and we can directly substitute $\lambda_i$ for $u_i$ as dictated by Eq.~(\ref{eq:symbol2realA}) and move on to evaluating the whole Dyck path. Let's start with a special case to make the exposition easy to follow.

\subsubsection*{Dyck paths $\euD(k,0,0)$}

The product representing the ascending portion of the highest Dyck path $\hiD(k,0,0)$ is $\prod_{i=1}^{k/2}\la_{i}$. The lowest Dyck path corresponds to $\la^{k/2}_{1}$. We can visualize the products obtained from the ascending segments of all Dyck paths in a different kind of diagram in  Fig.~\ref{fig:DyckStairs}~(a) reflecting the example $k=14$  in Fig.~\ref{fig:maxminDyck}~(a). The solid stair path is the highest product and the dotted line is the lowest one.  Any product from the ``wedge'' between the two extremal cases comes from an allowed Dyck path iff the indices for any consecutive pair $\la_i\la_j$ of a product satisfy
\begin{equation}\label{eq:allowedIndices}
  i-j\leq1.
\end{equation}
Using the fact that (i) the number of ascending segments  is equal to the number of descending  segments  (and so the number of $\la$ products evaluating the ascending sections is equal to the number of $\mu$ products) and (ii) their indices are identical as a consequence of~(\ref{eq:TheRule}), the sum over all evaluated Dyck paths reads
\begin{equation}\label{eq:mainSum1}
    \sum_{m_j=j}^{m_{j+1}}\la_{m_j-j+1}\mu_{m_j-j+1}\Bigg[
    \hdots\bigg[\sum_{m_2=2}^{m_{3}}\la_{m_2-1}\mu_{m_2-1}\bigg[\sum_{m_1=1}^{m_{2}}\la_{m_1}\mu_{m_1}\bigg]\bigg]\Bigg].
\end{equation}
We will see in Examples~\ref{exa:k2and3and5and100},~\ref{exa:2modesqueez} and~\ref{exa:3modescubic} and further in Sec.~\ref{sec:complex} that this $j$-multiple recursive  sum can be explicitly evaluated (and very fast for large $j$ for that matter). Note that the sum's upper bound does not change: $m_2=m_3=\hdots m_{j+1}\equiv k/2$ but for obvious reasons we had to introduce distinct dummy indices. Moreover, the number of sums in the recursive formula is also $j=k/2$. So indeed, the  outermost sum is not a sum at all since $\sum_{m_j=j}^{m_{j+1}}\la_{m_j-j+1}\mu_{m_j-j+1}=\la_1\mu_1$ as visualized by the overlapping segment of the solid and dotted line  in Fig.~\ref{fig:DyckStairs}~(a). Perhaps it is time for an example.
\begin{figure}
    \centering
  \includegraphics{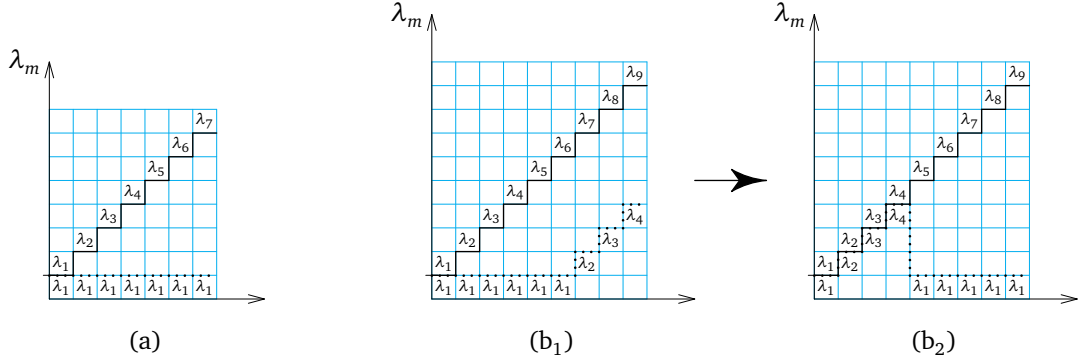}
   \caption{The $x$ axes count the number of ascending segments. The solid stair-like path and the dotted line on the left represent the evaluated ascending sections of the highest Dyck path $\hiD(14,0,0)$ and lowest Dyck path $\loD(14,0,0)$, respectively, that are depicted in Fig.~\ref{fig:maxminDyck}~(a).  In the middle plot we capture the same situation for $\euD(14,0,4)$ from Fig.~\ref{fig:maxminDyck}~(b). In the right plot we reordered the evaluated ascending segments of the lowest Dyck path as suggested by the ascending evaluated product for $\loD(14,0,4)$ in Eq.~(\ref{eq:loDykeReordered}).}
   \label{fig:DyckStairs}
\end{figure}
\begin{exa}\label{exa:explicitrecursion}
  Let $k=8$ and this case can be illustrated on a staircase diagram like the one in Fig.~\ref{fig:DyckStairs}~(a). In the present case the diagram ends with the stair where $\la_4$ lies. We can recreate Eq.~(\ref{eq:mainSum1}) from ``inside'' by following the condition~(\ref{eq:allowedIndices}) for $\la_i\la_j$ and the items (i) and (ii) above. Therefore, the innermost sum reads $\la_1\mu_1+\la_2\mu_2+\la_3\mu_3+\la_4\mu_4$ (see the horizontal column in  Fig.~\ref{fig:DyckStairs}~(a) between $\la_1$ on the dotted line and $\la_4$ on the stair). The sum is multiplied by the components of the preceding column ($\{\la_j\mu_j\}_{j=1}^3$) but it depends on $j$ (as~(\ref{eq:allowedIndices}) dictates) how many summands of the innermost sum  are allowed to be multiplied -- for $\la_j\mu_j$ the previous sum goes from 1 up to $j+1$:
    \begin{equation}\label{eq:explicitRecursionk8FirstIter}
      \begin{aligned}%\label{}
      &\ \la_3\mu_3(\la_1\mu_1+\la_2\mu_2+\la_3\mu_3+\la_4\mu_4) \\
        +&\ \la_2\mu_2(\la_1\mu_1+\la_2\mu_2+\la_3\mu_3) \\
        +&\ \la_1\mu_1(\la_1\mu_1+\la_2\mu_2).
      \end{aligned}
    \end{equation}
  Moving to the next level, only two values ($\la_1\mu_1$ and $\la_2\mu_2$) are possible. Hence by repeating the exact same rules, the third inner sum reads
    \begin{equation}%\label{}
      \begin{aligned}%\label{}
        &\ \la_2\mu_2\la_3\mu_3(\la_1\mu_1+\la_2\mu_2+\la_3\mu_3+\la_4\mu_4) \\
        +&\ \la_2\mu_2\la_2\mu_2(\la_1\mu_1+\la_2\mu_2+\la_3\mu_3) \\
        +&\ \la_2\mu_2\la_1\mu_1(\la_1\mu_1+\la_2\mu_2)\\
      \end{aligned}
    \end{equation}
  together with
    \begin{equation}%\label{}
      \begin{aligned}%\label{}
        &\ \la_1\mu_1\la_2\mu_2(\la_1\mu_1+\la_2\mu_2+\la_3\mu_3) \\
        +&\ \la_1\mu_1\la_1\mu_1(\la_1\mu_1+\la_2\mu_2)
      \end{aligned}
    \end{equation}
  obtained from the second and third line of~(\ref{eq:explicitRecursionk8FirstIter}) by multiplying by $\la_1\mu_1$.  Finally, the last sum (not being a sum since the only possibility is $\la_1\mu_1$) terminates the recursion:
    \begin{equation}\label{eq:explicitRecursionk8}
      \begin{aligned}%\label{}
        &\ \la_1\mu_1\la_2\mu_2\la_3\mu_3(\la_1\mu_1+\la_2\mu_2+\la_3\mu_3+\la_4\mu_4) \\
        +&\ \la_1\mu_1\la_2\mu_2\la_2\mu_2(\la_1\mu_1+\la_2\mu_2+\la_3\mu_3) \\
        +&\ \la_1\mu_1\la_2\mu_2\la_1\mu_1(\la_1\mu_1+\la_2\mu_2)\\
        +&\ \la_1\mu_1\la_1\mu_1\la_2\mu_2(\la_1\mu_1+\la_2\mu_2+\la_3\mu_3) \\
        +&\ \la_1\mu_1\la_1\mu_1\la_1\mu_1(\la_1\mu_1+\la_2\mu_2).
    \end{aligned}
  \end{equation}
  By revising the number of summands (equal to 14), which counts the number of Dyck paths $\euD(k,0,0)$, we may verify that it agrees with Eq.~(\ref{eq:Catalan}) for $k=8$.
\end{exa}

\subsubsection*{Dyck paths $\euD(k,0,\d_2)$}

As already mentioned, the strategy to list all possible Dyck paths $\euD(k,0,\d_2)$~is~similar to the case $\d_2=0$. Starting from the highest Dyck path we swap our way to the lowest Dyck path. Instead of doing it explicitly and encountering confusing expressions like those in Eqs.~(\ref{eq:FirstIter}) and~(\ref{eq:FirstIter2}), we better use the rule Eq.~(\ref{eq:TheRule}). But this time the situation seems to be more complicated. The descending section of any Dyck path is shorter than the ascending one whenever $\d_2>0$. So how do we know what $u_i$ and $d_i$ belong to the same Dyck path to properly evaluate the path according to Eq.~(\ref{eq:typeChange})? Because $\euD(k,0,\d_2)$ is confined to the positive quadrant, the descending section of the lowest Dyck path $\loD(k,0,\d_2)$ contains $(k-\d_2)/2$ of $d_1$ segments
\begin{equation}\label{eq:loDyke}
\underbrace{d_1\hdots\hdots d_1}_{k-\d_2\over2},
\end{equation}
whose evaluation (following Eq.~(\ref{eq:typeChange})) equals the product
\begin{equation}\label{eq:loDykeevaluated}
\mu_{1}^{k-\d_2\over2}.
\end{equation}
The highest Dyck path $\hiD(k,0,\d_2)$ is always of the form
\begin{equation}\label{eq:hiDyke}
\underbrace{d_{\d_2+1}\dots\dots d_{k+\d_2\over2}}_{k-\d_2\over2}\underbrace{u_{k+\d_2\over2}\dots\dots u_1}_{k+\d_2\over2}
\end{equation}
and the number of $d_i$ segments correctly coincides with~(\ref{eq:loDyke}). When~(\ref{eq:hiDyke}) is evaluated, we get the expression
\begin{equation}\label{eq:hiDevaluated}
\prod_{i=1}^{k-\d_2\over2}\mu_{\d_2+i}\prod_{i=1}^{k+\d_2\over2}\la_{i}.
\end{equation}
As previously described, to generate all Dyck paths we apply the rule Eq.~(\ref{eq:TheRule}) on the highest Dyck path Eq.~(\ref{eq:hiDyke}) and end up with the lowest one (Eq.~(\ref{eq:loDyke})). But the rule in Eq.~(\ref{eq:TheRule}) treats the ascending and descending segments (more precisely, the indices of their letters $u_i$ and $d_i$) in the same way. So it treats the indices of $\la_i$ and $\mu_i$ in the same way as well. As a consequence, only the underbraced portion of the ascending segment product taken from Eq.~(\ref{eq:hiDevaluated})
$$
\prod_{i=1}^{k+\d_2\over2}\la_{i}=\la_{1}\times\hdots\times\la_{\d_2}\underbrace{\la_{\d_2+1}\times\hdots\times\la_{(k+\d_2)/2}}_{k-\d_2\over2}
$$
will change as we descend to the lowest Dyck path. After  the lowest Dyck path is evaluated, its value must necessarily be equal to the product
\begin{equation}\label{eq:loDykeReordered}
\underbrace{\mu_{1}^{k-\d_2\over2}}_{\quad(\mbox{\scriptsize I})\quad}\underbrace{\la_{1}^{k-\d_2\over2}}_{\quad(\mbox{\scriptsize II})\quad}
\underbrace{\prod_{i=1}^{\d_2}\la_{i}}_{\quad(\mbox{\scriptsize III})\quad}.
\end{equation}
The $\mu_1$ product  denoted (I) is the evaluated descending section, product (II) corresponds to the part of the ascending section whose number $(k-\d_2)/2$ must be equal to product (I) and product (III) is the rest of the evaluated ascending section. We can  pictorially see what is happening  in Figs.~\ref{fig:DyckStairs}~($\mbox{b}_1$) and~($\mbox{b}_2$) on the example of $\euD(14,0,4)$ from Fig.~\ref{fig:maxminDyck}~(b).  The solid and dotted staircase in the middle plot is the evaluated ascending segment of the highest and lowest Dyck path as can be directly deduced from Fig.~\ref{fig:maxminDyck}~(b) using Eq.~(\ref{eq:symbol2realA}). Similarly to $\euD(k,0,0)$ exemplified in the left panel, all allowed evaluated ascending sections lie between the two extremal cases. An ascending section is allowed iff   the indices satisfy $i-j\leq1$ for any consecutive pair $\la_i\la_j$ of a product. It is, however, extremely advantageous to reorder the lowest evaluated section according to the prescription derived in Eq.~(\ref{eq:loDykeReordered}). Then we are able to see the product of (II) and (III) from Eq.~(\ref{eq:loDykeReordered}) in the right panel of Fig.~\ref{fig:DyckStairs} as the dotted path. More importantly, all products contain expression (III) and so it can be factored out. In the example, this happens to be $\la_1\la_2\la_3\la_4$ because $\d_2=4$.  The reason we underwent this laborious transformation\footnote{The transformation from $(\mbox{b}_1)$ to $(\mbox{b}_2)$ can  also be understood as a consequence of $\la_i$ being scalars and therefore commuting. Then, in the right panel, the allowed evaluated ascending sections are  suitably reordered allowed sections from the middle panel.} is to find out the summing formula over all evaluated Dyck paths $\euD(k,0,\d_2)$. With the help of Eq.~(\ref{eq:loDykeReordered}) and following the strategy leading to Eq.~(\ref{eq:mainSum1}) we are able to read off the  summing formula from Fig.~\ref{fig:DyckStairs}~($\mbox{b}_2$):
\begin{equation}\label{eq:mainSum2}
    \prod_{i=1}^{\d_2}\la_{i}\Bigg(
    \sum_{m_j=j}^{m_{j+1}}\la_{m_j-j+1}\mu_{m_j-j+1}\Bigg[
    \hdots\bigg[\sum_{m_2=2}^{m_{3}}\la_{m_2-1}\mu_{m_2-1}\bigg[\sum_{m_1=1}^{m_{2}}\la_{m_1}\mu_{m_1}\bigg]\bigg]\Bigg]
    \Bigg).
\end{equation}
Eq.~(\ref{eq:mainSum2}) is a modification of Eq.~(\ref{eq:mainSum1}) in a transparent way: the total number of sums is $j=(k-\d_2)/2$ -- this is where the solid and dotted paths in Fig.~\ref{fig:maxminDyck}~($\mbox{b}_2$) do not overlap. The overall multiplicative constant is the overlapping portion of the diagram where the values of $\la_i$ and $\mu_i$ are fixed. The upper bounds in~(\ref{eq:mainSum2}) have again the same value $m_2=\hdots=m_{j+1}=(k+\d_2)/2$.

Eqs.~(\ref{eq:mainSum1}) and (\ref{eq:mainSum2}) can be used to calculate Eq.~(\ref{eq:ExpFinalExpansion}). This is where the true power of our method is exposed: summing over all evaluated Dyck paths is greatly facilitated whenever the action of $A$ and $\Adg$ in Eqs.~(\ref{eq:LadderOps}) is sufficiently ``nice''. But polynomials are always summable (via Bernoulli's formula) and so the iterative sums in Eqs.~(\ref{eq:mainSum1}) and (\ref{eq:mainSum2}) are all closed expressions.

To compare our algorithm with other analytical approaches it is fair to say that for a very special case of $A=a^d$ the expression $\Ncal\big[(\adg^d+a^d)^k\big]$ can be used for the same task. The methods to normal order the expression were developed in~\cite{blasiak2005boson,blasiak2005combinatorics}. A different kind of combinatorial insight was used by the authors (more in the spirit of~\cite{katriel73}) with the help of generalized Stirling numbers. The results were subsequently used for the calculation of the expectation value $\ave{(\adg^d+a^d)^k}_{\ket{\a}}$ for a coherent state $\ket{\a}$. From the conceptual point of view, normally ordered expressions are satisfactory since they are essentially closed and hold in the strong sense, that is, without acting on any state. But that may be actually a mild disadvantage because it still has to act on a state and the resulting formula can become further complicated. What is even perhaps less clear is the computational complexity of how effectively it can be evaluated for $k\gg1$ in actual calculations. But irrespective of that, here we aim at multiboson operators $A$ and $\Adg$ (constituting physically interesting Hamiltonians), where to the author's best knowledge, analytical results are essentially  non-existent.  This will be further used to evaluate the action of the isometry Eq.~(\ref{eq:generalEvolutionOp}) expanded as a Taylor series to high order in $k$.

The only (partial) exception seems to be~\cite{mansour2008normal} where, however, two boson modes are coupled such that $[a_1,a^\dg_2]=1$. This is not the behavior of ordinary bosons considered here and in quantum field theory or quantum optics.

\subsection{Complexity estimation}\label{sec:complex}

Let's analyze the complexity of calculating Eqs.~(\ref{eq:mainSum1}) and (\ref{eq:mainSum2}) for the boson operator $A$ being a product of $d$ creation and annihilation operators. This will give us an idea of how tractable our method is. We will focus on~(\ref{eq:mainSum1}) as the worst case scenario and denote the sums' upper bounds $m$. From the general action of a boson creation and annihilation operator in Eqs.~(\ref{eq:bosonSteps}) we find that $\la_{m_1}\mu_{m_1}=f^{(d)}$ is a polynomial of order $d$ (written symbolically). Then the innermost sum of~(\ref{eq:mainSum1})
\begin{equation}\label{eq:FirstRecursion}
\sum_{m_1=1}^{m}\la_{m_1}\mu_{m_1}\equiv\sum_{m_1=1}^{m}f^{(d)}=f^{(d+1)}
\end{equation}
is a  polynomial of order $d+1$ and so the cost of calculating the sum is $\Ocal(dm)$. To evaluate the following sum, the polynomial $f^{(d+1)}$ is multiplied by  $\la_{m_2}\mu_{m_2}$ and this is again a polynomial of order $d$: $\la_{m_2}\mu_{m_2}=f^{(d)}$. The polynomial is summed over $m_2$ but only to $m-1$. In Eq.~(\ref{eq:mainSum1}) it is the lower bound that changes by one so it is the same number of operations. The whole recursive formula can then be symbolically depicted as follows
\begin{equation}\label{eq:symbPolys}
  \left.\begin{aligned}
    f^{(d)}  &\overset{\sum^{m}_1\limits}{\longrightarrow}  f^{(d+1)}  \\
    f^{(d+1)}f^{(d)}  &\overset{\sum_1^{m-1}\limits}{\longrightarrow} f^{(2d+2)}\\
    f^{(2d+2)}f^{(d)} & \overset{\sum_1^{m-2}\limits}{\longrightarrow} f^{(3d+3)}\\
    &\ \dots
\end{aligned}\right\rbrace\quad m\mbox{\ times.}
\end{equation}
We first find the order of the resulting polynomial. From the RHS of (\ref{eq:symbPolys}) we deduce that the polynomial after the last iteration $m=k/2$ is of order $m(d+1)=k/2(d+1)$ (recall that since $\d_2=0$, $k$ is always even). For the complexity estimation we observe  that the $i$-th recursive sum on the LHS of~(\ref{eq:symbPolys}) is a sum over $m-i+1$ summands and a polynomial of order $di+i-1$. Henceforth
\begin{equation}\label{eq:polyComplexity}
  \sum_{i=1}^m(m-i+1)(di+i-1)=\frac{1}{6} \left((d+1)m^3+3 d m^2+2 d m-m\right)
\end{equation}
and so the recursive calculation roughly scales as $\Ocal(k^3)$.

This is the main advantage of our method. It provides a tremendous (exponential) speed-up for calculating $(\Adg\pm A)^k\psi$. There are $2^k$ summands after expanding $(\Adg\pm A)^k$ and so their number grows exponentially with $k$. This is prohibitively slow to compute even for a small $k$ and it is perhaps obvious that a similar obstacle remains even after we take into account the truncated sum in Eq.~(\ref{eq:ExpFinalExpansion}).   In more detail, one finds
\begin{equation}\label{eq:estimate}
%\mathop{\sum_{\d_2\in\{2m+(k\bmod{2})}}_{|0\leq m\leq k\}}G(k,0,\d_2)=\binom{k}{\lfloor{k\over2}\rfloor}=\Ocal\big(2^{k-\log_2{\sqrt{k}}}\big),
\mathop{\sum_{\substack{\ \d_2=0,2,\dots,k\ \ \mbox{\scriptsize{or}}\\\d_2=1,3,\dots,k-1}}}G(k,0,\d_2)=\binom{k}{\lfloor{k\over2}\rfloor}=\Ocal\big(2^{k-\log_2{\sqrt{k}}}\big),
\end{equation}
where the Stirling formula for the factorial can be used to get the asymptotic estimate of the RHS. This is the most favorable case since for $\d_1>0$ the number of summands in Eq.~(\ref{eq:ExpFinalExpansionGeneral}) is always bigger for fixed $k$ and $\d_2$. By setting $\d_1=k$ we get the worst-case scenario where the growth is simply $\Ocal(2^k)$ since none of the summands from $(\Adg+A)^k$ acting on a state $\psi^{(k)}$ will vanish (see the outer sum of Eq.~(\ref{eq:ExpFinalExpansionGeneral})). The largest number of Dyck paths in the worst case scenario happens when $\d_1=\d_2$ for $k$ even as long as $\d_1\geq k/2$. In this case  Eq.~(\ref{eq:mostGeneralDyckCounting}) tells us that the number of Dyck paths is $G(k,\d_1,\d_1)=\binom{k}{k/2}$. The Dyck paths are then confined to a diamond-like shape illustrated in Fig.~\ref{fig:pathstairs}~(a) for $\d_1=\d_2=k=8$.
\begin{figure}[t]
  \centering
  \includegraphics{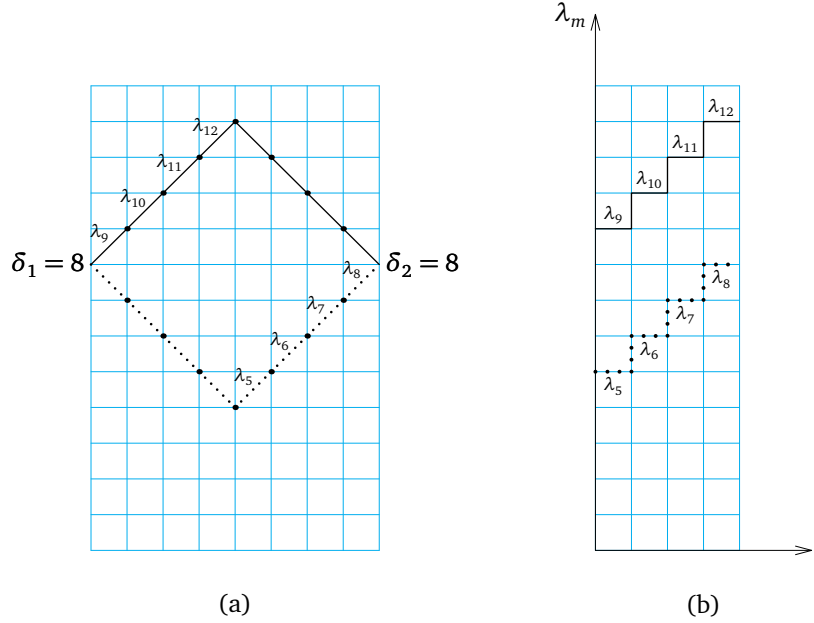}
   \caption{All Dyck paths $\euD(8,8,8)$ are confined to the diamond on the left where the ascending segments of the highest and lowest (dotted) Dyck path are evaluated by $\la_i$. In panel (b) we see the corresponding staircase diagram to illustrate the complexity of calculating $(\Adg+A)^k\psi^{(\d_1)}$ in the worst case scenario that occurs for a given $k$ when $\d_1=\d_2$ and $\d_1\geq k/2$ ($k$ even). Here we set $\d_1=k$.}
   \label{fig:pathstairs}
\end{figure}
In Fig.~\ref{fig:pathstairs}~(b) we see in the corresponding staircase diagram that unlike the diagram in Fig.~\ref{fig:DyckStairs}~(a),  the ascending segments are not approaching. Consequently, the number of recursive sums ($m=k/2+1=5$ in~(\ref{eq:FirstRecursion}) in this case) in each step is constant. Following the estimate as in Eq.~(\ref{eq:symbPolys}) we arrive at
\begin{equation}%\label{eq:polyComplexity}
  (m+1)\sum_{i=1}^m(di+i-1)={1\over2}\big((d+1)m^3+2dm^2-m(1-d)\big).
\end{equation}
We see that the calculation is nearly as effective as for in~(\ref{eq:polyComplexity}).

Let's go back to the $\d_1=0$ case as that will be our source of examples to come. The number of sets of Dyck paths with  the same value of $\d_2$ grows linearly with $k$ in the expression $(\Adg+A)^k$. Because we have just shown that the sum over these ``$\d_2$-sets'' in Eq.~(\ref{eq:mainSum1}) goes as $\Ocal(k^3)$, the calculation of (\ref{eq:ExpFinalExpansion}) becomes tractable. But this also means that the calculation of
\begin{equation}\label{eq:generalEvolutionOpExpanded}
  V\psi^{(0)}=\exp{\big[r(A^\dg-A\big)\big]}\psi^{(0)}\simeq\sum_{k=0}^K{r^k\over k!}(A^\dg-A\big)^k\psi^{(0)}
\end{equation}
is polynomial in $K$ ($\psi^{(0)}$ is the Minkowski vacuum $\ket{0}$). The number of summands in~(\ref{eq:generalEvolutionOpExpanded}) grows linearly with $K$ but as will be discussed in Example~\ref{exa:k2and3and5and100} for $k=100$, even this can be significantly reduced.

Let's  demonstrate the power of Eqs.~(\ref{eq:mainSum1}) and~(\ref{eq:mainSum2}).
\begin{exa}\label{exa:k2and3and5and100}
As a generic single-mode example we will calculate $(\Adg+A)^k\ket{0}$, where $A=a^3$, for several values of $k$.
  \begin{enumerate}%[label=\bfseries{\roman*})]
    \item[\fbox{$k=2$}] In this case  only two possibilities exist: $\d_2=0,2$. For $\d_2=0$ the number summands $j$ in Eq.~(\ref{eq:mainSum1}) equals $k/2=1$ and so there is nothing to sum -- only the outermost sum survives. But this is not a sum at all as previously discussed because there is only one Dyck path $\euD(2,0,0)$. It remains to calculate $\mu_1$ and $\la_1$ with the help of Eqs.~(\ref{eq:bosonSteps}).         Hence, by denoting $\psi^{(k-1)}=\ket{3(k-1)}$ we get from Eqs.~(\ref{eq:LadderOps})
        \begin{equation}\label{eq:exak2Coefs}
          \la_k=\mu_k=\sqrt{3k(3k-1)(3k-2)}
        \end{equation}
        and so $\mu_1=\la_1=\sqrt{3!}$.

        For $\d_2=2$ the situation is similar. The number of sums in Eq.~(\ref{eq:mainSum2}) is $j=(k-\d_2)/2=0$ so only the  overall factor $\la_1\la_2$ is present. From Eq.~(\ref{eq:exak2Coefs}) we get $\la_2=\sqrt{4\times5\times6}$. The final answer is
        $$
        \big(\Adg+A\big)^2\ket{0}=\la_1\mu_1\ket{0}+\la_1\la_2\ket{6}=3!\ket{0}+\sqrt{6!}\ket{6}.
        $$
        This can be readily confirmed by directly calculating $\big(a^3+\adg^3\big)^2\ket{0}$.
    \item[\fbox{$k=3$}] For $\d_2=1$ there are two Dyck paths $\euD(3,0,2)$. Let's evaluate this case only. The number of summands in Eq.~(\ref{eq:mainSum2}) is now one ($(k-\d_2)/2=1$) and its upper bound is $m_2=(k+\d_2)/2=2$. But to start showing the potential of the Dyck path approach for general $k$ we will not set $m_2=2$ yet. Instead, we calculate
        \begin{equation}\label{eq:firstSumk2}
          \sum_{k=1}^{m_2}3k(3k-1)(3k-2)=\frac{3}{4}m_2(m_2+1)(3m_2-2)(3m_2+1)
        \end{equation}
        and take the liberty of setting $m_2=2$ in the evaluated sum after we write the final answer given by Eq.~(\ref{eq:mainSum2}):
        $$
        \frac{3\sqrt{3!}}{4}m_2(m_2+1)(3m_2-2)(3m_2+1)=\sqrt{3!}(3!+4\times5\times6)=\sqrt{3!}\,126.
        $$
        This is the first step to achieve a closed expression even for a large $k$ to be used in the next two examples. The calculation again agrees with the $\ket{3}$ coefficient of
        $$
        \big(\adg^3+a^3\big)^3\ket{0}=\sqrt{3!}(3!+120)\ket{3}+\sqrt{9!}\ket{9}.
        $$
    \item[\fbox{$k=5$}] Here we will illustrate the appearance of one nested sum in Eq.~(\ref{eq:mainSum2}) for $\d_2=1$ coming from $j=2$. The recursive summation results in
        \begin{align}\label{eq:resultExamplek5}
           & \frac{9}{1120} (m_3-1)m_3(m_3+1)(m_3+2)\nn \\
          \times & \big(2835 m_3^4-5130 m_3^3-3915 m_3^2+8034 m_3-808\big),
        \end{align}
        where now we set $m_3=3.$ The final answer for $\d_2=1$ is therefore $\sqrt{3!}\,76356$. This can be verified by hand but it is already rather lengthy.
    \item[\fbox{$k=100$}] To show the strength of our method we calculate $\big(\Adg+A\big)^{100}\ket{0}\equiv\big(\adg^3+a^3\big)^{100}\ket{0}$ for $\d_2=0$. The number of Dyck paths  given by Eq.~(\ref{eq:Catalan}) is
        $$
        1\,978\,261\,657\,756\,160\,653\,623\,774\,456.
        $$
        The brute-force calculation is out of the question on today's computers.

        Using our approach we see that Eq.~(\ref{eq:mainSum1}) contains 50 recursive sums. It takes roughly 20 seconds to obtain an analytical result in Mathematica on a single-core average laptop. The result (the equivalent of (\ref{eq:resultExamplek5})) is a polynomial of order 200 whose printout would take 24 A4~pages. The order coincides with the derivation below Eq.~(\ref{eq:symbPolys}) for $d=3$ (note $k/2(d+1)=200$). Given the recursive character of Eq.~(\ref{eq:mainSum1}), the code used to generate the result is a one-liner. To calculate the complete expression $\big(\adg^3+a^3\big)^{100}\ket{0}$ we, of course, have to calculate the remaining cases for the rest of the $\d_2$'s ($\d_2=2,4,\dots,100$) given by Eq.~(\ref{eq:mainSum2}). The number of recursions is $j=(k-\d_2)/2$ but these recursions are already contained in the largest calculation for $\d_2=0$. The only thing that differs for each $\d_2$ is the sums' upper bound. Recall that we calculate the sum for an arbitrary upper bound $m_i$, where $2\leq i\leq{j+1}$, and plug the value $(k+\d_2)/2$ after the calculation. So we only need to save all~$j$ intermediate nested sums calculated for $\d_2=0$, plug $m_j=(k+\d_2)/2$ for the rest of $\d_2$ and multiply it by the overall factor $\prod_{i=1}^{\d_2}\la_i$ as required by Eq.~(\ref{eq:mainSum2}). Indeed, the first nested sum in the $k=100$ calculation is precisely Eq.~(\ref{eq:resultExamplek5}), where instead of $m_3$ we have $m_{50}$. So all the remaining $\d_2$ operations are computationally for free!
    \item[\fbox{$k=200$}] This case is to verify the cubic complexity scaling (for a fixed $d$) derived in Eq.~(\ref{eq:polyComplexity}) for the recursive sum Eq.~(\ref{eq:mainSum1}). It takes some 240 seconds to obtain an analytical result in Mathematica. This is indeed expected: $200^3/100^3=8$.  The result is a polynomial of order 400 whose printout fits 105 A4 pages.
\end{enumerate}
\end{exa}
Even without using the computational shortcut for $\d_2>0$ described in the $k=100$ case, we already got rid of the exponential time complexity by summing over all Dyck paths $\euD(k,0,\d_2)$.

The cubic boson expression from the previous examples can be thought as a Hamiltonian generating a unitary evolution operator, see Example~\ref{exa:decoup} for the exact formulation. But the following unitary is even more important: the two-mode squeezing operator. Its factorization is a routine procedure and so we can compare it with our approach.

\begin{exa}\label{exa:2modesqueez}
  Let's investigate $V\psi^{(0)}=\exp[r(a^\dg_1a^\dg_2-a_1a_2)]\psi^{(0)}$, where $\psi^{(0)}$ is a ground state coinciding with the Minkowski vacuum $\ket{0}$ annihilated by $a_i$.  We will compare the closed form of the amplitude  coefficients provided by
  \begin{equation}\label{eq:twomodeSqueezAnal}
    \exp[r(a^\dg_1a^\dg_2-a_1a_2)]\ket{0}={1\over\cosh{r}}\sum_{n=0}^\infty\tanh^{n}{r}\ket{n}_{1}\ket{n}_{2}
  \end{equation}
  with
  \begin{equation}\label{eq:twomodeSqueezTaylor}
    V\ket{0}\simeq\sum_{k=0}^K{r^k\over k!}(a^\dg_1a^\dg_2-a_1a_2)^k\ket{0}
  \end{equation}
  coming from~Eq.~(\ref{eq:generalEvolutionOpExpanded}).   The calculation reduces to $(a^\dg_1a^\dg_2-a_1a_2)^k\ket{0}\equiv(A-\Adg)^k\ket{0}$ for $0\leq k\leq K$. Using the methods developed in this paper we can afford to set $K$ very high and study the evolution for any fixed value of $r$ for which $V$ converges. From Eq.~(\ref{eq:bosonSteps}) we find
  $$
  \mu_k=\la_k=k
  $$
  by setting $\psi^{(k-1)}=\ket{k-1}_1\ket{k-1}_2$ in Eqs.~(\ref{eq:LadderOps}). It is perhaps clear that for any $k$, the admissible values of $\d_2$ tell us to what state $\ket{n}_{1}\ket{n}_{2}$ we calculate the amplitude contribution. So, for example, whenever $\d_2=0$, the contribution goes to the vacuum amplitude $\ket{0}$, the $\d_2=1$ contributions go to $\ket{1}_{1}\ket{1}_{2}$ etc., up to $\ket{k}_{1}\ket{k}_{2}$ since we know that $0\leq\d_2\leq k$ if $k$ is even (if it is odd we offset $\d_2$ by one, see Sec.~\ref{sec:DyckCount}).

  If we set $K=200$, all 200 calculations $(a^\dg_1a^\dg_2-a_1a_2)^k\ket{0}$ take roughly 250 seconds to complete \emph{without}  using further optimization described in the $k=100$ case of Example~\ref{exa:k2and3and5and100}. As a matter of fact, the result of the $k$-th expansion summand could have been used to speed up the calculation in all the exponents that followed. Neither this optimization was implemented. Just for comparison, considering only the highest contribution ($k=200$) to the vacuum amplitude ($\d_2=0$), they come from the immense number  $$
  C_{200}\approx2^{189}
  $$
  of Dyck paths, according to~Eq.~(\ref{eq:Catalan}).

  If the evolution operator $V$ was evaluated using its Taylor expansion, it would be necessary to assume $r\to0$ as only a few first expansion coefficients are possible to manage. Not in our case. For $K=200$ and a very high\footnote{The magnitude of $r$ is a relative notion especially because  $0\leq r<\infty$. Nevertheless, compared to what has been possible so far for the Taylor expansion, it is indeed high and following our performance analysis in Sec.~\ref{sec:complex}, we can now calculate any finite value just by sufficiently increasing the Taylor expansion. From a physical point of view it is also high. If translated to the quantum optical scenario, where $r$ is a squeezing parameter, then the squeezing value $r=1$ corresponds to the noise reduction $10\log_{10}{e^{2r}}\approx8.7\mbox{\ dB}$. This is roughly current state-of-the-art in optical experiments~\cite{PhysRevLett104251102}. If we use a special-relativistic comparison based on the local isomorphism $SU(1,1)/\bbZ_2\simeq SO(2,1)$~\cite{perelomov1986generalized,yurke1986}, then $2r=2$ is the rapidity corresponding approximately to $96.4\%$ of the speed of light.} value of the coupling constant $r=1$, we get very good agreement with $\tanh^n{r}/\cosh{r}$ up to $n=33$:
  \begin{equation}\label{eq:squeezComparion}
    a_{33}={\tanh^{33}{1}\over\cosh{1}}\simeq0.000081001\quad\mbox{vs.}\quad\widetilde{a}_{33}=0.0000799909.
  \end{equation}
  The lower coefficients starting with $n=30$ are essentially indistinguishable from the ones given by the analytical expression. By comparing the Taylor expansion of the analytical coefficients $a_{n}={\tanh^n{r}\over\cosh{r}}$ with $\tilde{a}_{n}$ we indeed see the perfect match. Also recall that the calculation of Eq.~(\ref{eq:twomodeSqueezTaylor}) is perturbative but analytical and the value of the coupling constant $r$ is inserted after the calculation. So the expansion is not needed to be recalculated for every numerical value of $r$.
\end{exa}
\begin{exa}\label{exa:3modescubic}
The evolution operator
\begin{equation}\label{eq:cubic}
  W=\exp{\big[r({a^\dg}^3-a^3)\big]}
\end{equation}
is supposed to describe the three-photon degenerate parametric amplifier. Some time ago it attracted a lot of attention by an observation made in~\cite{fisher1984impossibility} that the vacuum expectation value does not converge for $r>0$. The situation was clarified by Braunstein and McLachlan~\cite{braunstein1987generalized} who uncovered the immediate reason  for divergence and used an analytic continuation method based on Pad\'e approximants~\cite{baker1975essentials}  showing convergence for a finite interval of $r$. It is possible that the cause of the bad behavior of $W$ is the lack of a self-adjoint extension on the domain of interest of the Hamiltonian $H_{cub}=r\big({a^\dg}^3-a^3\big)$ as this is a subtle issue for unbounded operators. But to the author's knowledge this problem has not been clarified yet. What is clear is that the Taylor series is essentially useless but even here our method can be used. Following the footsteps of~\cite{braunstein1987generalized} we converted the Taylor expansion coefficients of $W\psi^{(0)}$ ($\psi^{(0)}$ is the Minkowski vacuum $\ket{0}$ and (\ref{eq:exak2Coefs}) was used in the calculation and the difference in the minus sign was taken into account as well) to the Pad\'e polynomial $[83/83]$ and arrived at an identical result. See Fig.~\ref{fig:cubic} and compare it with Fig.~2(a) in~\cite{braunstein1987generalized}.
\begin{figure}[t]
  \centering
  \includegraphics{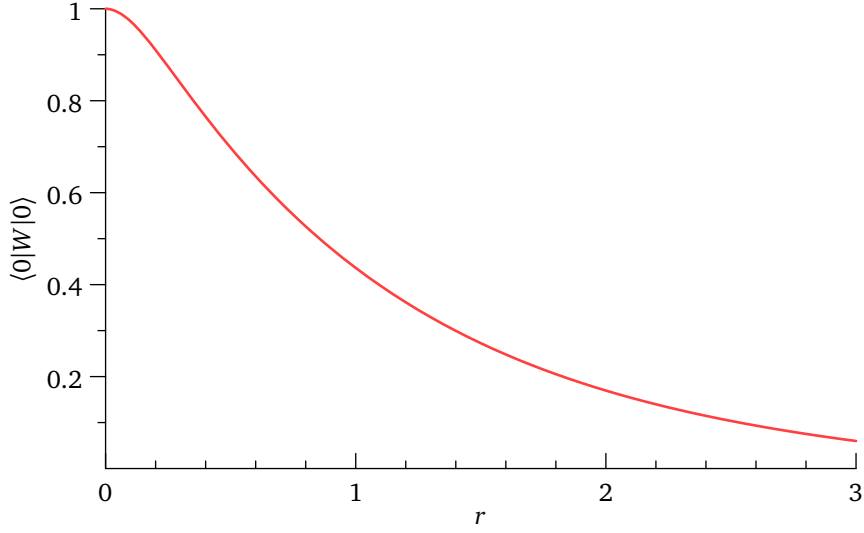}
   \caption{The vacuum expectation value for the cubic boson evolution operator $W=\exp{\big[r({a^\dg}^3-a^3)\big]}$ is plotted as a function of the coupling constant. It coincides with the result of~\cite{braunstein1987generalized} showing that our method can be useful even for the evolution operators whose Taylor series diverges.}
   \label{fig:cubic}
\end{figure}
\end{exa}
\begin{exa}\label{exa:decoup}
  For our final example of the developed method I invite the reader to study~\cite{bradler2015one}, where a recently developed model of a unitary black hole evaporation based on the trilinear boson Hamiltonian
   \begin{equation}\label{eq:triHam}
    H_{\mathrm{tri}}=r\big(ab^\dg c^\dg-\adg bc\big)
    \end{equation}
   is investigated. Similarly to the expression in Eq.~(\ref{eq:cubic}) the unitary operator $W=\exp{[H_{\mathrm{tri}}]}$ is not known to be easily factorizable but unlike the cubic evolution operator it exhibits no problems when it comes to the convergence of its Taylor series.
\end{exa}

\section{Conclusions}

In this work, we introduced a method of evaluating evolution operators for the interaction picture Hamiltonians of the form $H=r(\Adg-A)$, where $A$ is a multimode boson  monomial and an extensive class of sums of boson monomials and $r$ is a coupling constant.  The calculation of the evolution operator $V=\exp{[r(\Adg-A)]}$ is in general a difficult problem. Even though a factorization is always possible, for example, by virtue of the Zassenhaus formula or other decoupling techniques, a simple factorization into a small number of products  is available only if the operators $A$ and $\Adg$  satisfy favorable algebraic properties, such as vanishing commutators. If this is not satisfied, it may happen that there is no simple method left and the remaining techniques at our disposal are computationally demanding.

The Taylor expansion in $r$ with its subsequent action on a state of interest $\psi^{(0)}$ (most often a ground state of $A$ defined as $A\psi^{(0)}=0$) is certainly not the most efficient method as the number of  summands increases exponentially with the expansion order. The method developed here overcomes this problem and allows us to analytically calculate the action of the evolution operator $V$ to a very high order of its Taylor expansion and previously hardly accessible values of the coupling constant. This is possible due to an insight that a  combinatorial structure known as a Dyck path can be related to the action of boson monomials on a ground state $\psi$. It helps us to dramatically reduce the complexity of analytically calculating $(\Adg\pm A)^k\psi^{(0)}$ from $\Ocal\big(2^{k-\log_2{\sqrt{k}}}\big)$ to $\Ocal(dk^3)$  for all multimode boson monomials $A$ of the length $d$ and any ground state $\psi^{(0)}$. The boson expression is therefore not required to satisfy any special algebraic property except for the fact that they behave as ladder operators: $A\,\Adg\psi^{(p)}\propto\psi^{(p)}$. The state $\psi^{(p)}$ is any state from the tower of states generated by the repeated action of $\Adg$ on a ground state. The result is applicable for $V$ acting on a ground state  but the technique is equally efficient for all $\psi^{(p)}$. The  speed-up  is possible due to the existence of a summing recursive formula that can be explicitly evaluated for all such $A$ and $\Adg$ forming $(\Adg\pm A)^k$. Consequently, the complexity of calculating the evolution operator $V\psi^{(0)}$ is polynomial-time as well. All boson  evolution operators (their Taylor expansions) are expressed in terms of $A$ and $\Adg$  and we believe that our technique could be relevant for computational purposes in condensed matter, quantum field theory, quantum optics and other branches of quantum physics. We illustrated the method on a non-trivial example of $A$ and also verified it in the case of a two-mode squeezed operator, where a simple factorization procedure is well known and the degenerate cubic Hamiltonian, whose Taylor series is known to diverge for any value of the coupling constant $r$ and has to be analytically continued using, for instance, the Pad\'e approximants.

%\section*{References}

\bibliographystyle{unsrt}

%\bibliography{dyck}

\end{document}